\documentclass[11pt]{article}

\usepackage[utf8]{inputenc}
\usepackage{float}

\usepackage{amsmath,amsfonts,amsthm,amssymb,color}
\usepackage[dvipsnames,svgnames,table]{xcolor}
\definecolor{darkgreen}{rgb}{0.0,0,0.9}
\usepackage{mathtools}
\usepackage{authblk}
\usepackage{fullpage}
\usepackage{todonotes}
\usepackage{parskip}
\usepackage{comment}
\usepackage{tikz}
\usepackage{bbm}
\usepackage{dsfont}
\usepackage[sc]{mathpazo}
\usepackage[basic]{complexity}
\usepackage{algorithm2e}
\usepackage[colorlinks=true,citecolor=OliveGreen,linkcolor=BrickRed,urlcolor=BrickRed,pdfstartview=FitH]{hyperref}
\usepackage[capitalize,nameinlink]{cleveref}
\usepackage{tcolorbox}

\newtcolorbox{wbox}
{
	colback  = white,
}

\SetKwInOut{Input}{Input}
\SetKwInOut{Output}{Output}
\SetKwFunction{Uncross}{\textsc{Uncross}}
\SetKwFunction{MergeUncross}{\textsc{MergeUncross}}
\SetKwFunction{ConnectedComponents}{\textsc{ConnectedComponents}}
\SetKwFunction{PerfectMatching}{\textsc{PerfectMatching}}
\SetKwBlock{InParallel}{in parallel do}{end}
\SetKwFor{ParallelFor}{for}{in parallel do}{end}

\theoremstyle{definition}
\newtheorem{theorem}{Theorem}
\newtheorem{lemma}{Lemma}
\newtheorem{corollary}{Corollary}
\newtheorem{definition}{Definition}

\newtheorem{observation}[theorem]{Observation}

\newtheorem{claim}[theorem]{Claim}

\title{On the Core of the $b$-Matching Game}

 \author[1]{Rohith Reddy Gangam
 }
 \author[1]{Shayan Taherijam}
 \author[1]{Vijay V. Vazirani}

 \affil[1]{University of California, Irvine}

\date{}

\begin{document}
    \maketitle


\begin{abstract}

The core is a quintessential solution concept for profit sharing in cooperative game theory. An imputation allocates the worth of the given game among its agents. The imputation lies in the core of the game if, for each sub-coalition, the amount allocated to its agents is at least the worth of this sub-coalition. Hence, under a core imputation, each of exponentially many sub-coalitions gets satisfied. 

The following computational question has received much attention: Given an imputation, does it lie in the core? Clearly, this question lies in co-NP, since a co-NP certificate for this problem would be a sub-coalition which is not satisfied under the imputation. This question is in P for the assignment game \cite{Shapley1971assignment} and has been shown to be co-NP-hard for several natural games, including max-flow \cite{Fang2002computational} and MST \cite{Faigle1997complexity}. The one natural game for which this question has remained open is the $b$-matching game when the number of times an edge can be matched is unconstrained; in case each edge can be matched at most once, it is co-NP-hard \cite{biro2018stable}.

At the outset, it was not clear which way this open question would resolve: on the one hand, for all but one game, this problem was shown co-NP-hard and on the other hand, proximity to the assignment problem and the deep structural properties of matching could lead to a positive result. In this paper, we show that the problem is indeed co-NP-hard. 

\end{abstract}

    

\section{Introduction}

The core is a quintessential solution concept for profit sharing in cooperative game theory. An imputation allocates the worth of the given game among its agents. The imputation lies in the core of the game if, for each sub-coalition, the amount allocated to its agents is at least the worth of this sub-coalition. Hence, under a core imputation, each of exponentially many sub-coalitions gets satisfied. 

The following computational question has received much attention: Given an imputation, does it lie in the core? Clearly, this question lies in co-NP, since a co-NP certificate for this problem would be a sub-coalition which is not satisfied under the imputation. This question is in P for the assignment game \cite{Shapley1971assignment} and has been shown to be co-NP-hard for several natural games, including max-flow \cite{Fang2002computational} and MST \cite{Faigle1997complexity}. The one natural game for which this question has remained open is the $b$-matching game when the number of times an edge can be matched is unconstrained; in case each edge can be matched at most once, it is co-NP-hard \cite{biro2018stable}.

At the outset, it was not clear which way this open question would resolve: on the one hand, for all but one game, this problem was shown co-NP-hard and on the other hand, proximity to the assignment problem and the deep structural properties of matching could lead to a positive result. In this paper, we show that the problem is indeed co-NP-hard. 

\textbf{Results:}
We show the following results in our paper.
\begin{enumerate}
    \item Deciding if an imputation is in the core of edge-unconstrained $b$-matching game is co-NP-complete.
    \item We extend this hardness result to show that recognizing if an imputation is the leximin 
   (/leximax) in the core is also co-NP-complete. This also shows that the computation of these imputations is also NP-hard.
    \item We complement these hardness results by giving a complete characterization of the core of $b$-matching games on star graphs. This gives a polynomial time algorithm to decide if an imputation is in the core of these games.

\end{enumerate}


    \section{Related Works}

Graph-based cooperative games and the properties of core imputations within them have been extensively studied. The seminal work of \cite{Shapley1971assignment} established that in the assignment game, the core is precisely characterized by the set of optimal dual solutions. A natural extension of the assignment game is the $b$-matching game, where each vertex can be matched multiple times up to a predefined capacity.

There are two versions of the $b$-matching game. \cite{b_matching_nucleolus} distinguish them as simple and non-simple $b$-matching games. The simple $b$-matching game is the edge-constrained version where each edge is to be matched at most once. \cite{biro2018stable} demonstrated that verifying whether an imputation belongs to the core in this setting is co-NP-complete. Our work focuses on the non-simple $b$-matching game, the edge-unconstrained version, where edges can be used multiple times. These games are often referred to as \textit{transportation games} (see \cite{Transportation_games}), and we establish a similar hardness result in this setting. Notably, \cite{sotomayor1992multiple}) proved that the core of transportation games is always non-empty.

Beyond $b$-matching games, another well-studied class of cooperative games is the \textit{minimum-cost spanning tree (MST) games}. \cite{GranotHuberman1981} introduced MST games and showed that their core is always non-empty. Several methods have been developed to compute core imputations, including Bird’s rule(\cite{Bird1976}) and fairness-based approaches (\cite{Kar}, \cite{feltkamp1994irreducible}, \cite{bergantinos2007fair}, \cite{trudeau2012new}). However, verifying whether an imputation belongs to the core of an MST game is co-NP-hard (\cite{Faigle1997complexity}).

A closely related class of cooperative games is the max-flow game, introduced by \cite{Kalai1982totally}, where the core is also always non-empty. Yet, similar to MST games, testing whether an imputation belongs to the core is co-NP-hard (\cite{Fang2002computational}).

Transportation games share structural properties with these games. \cite{Transportation_games} demonstrated that any optimal dual solution of the linear programming (LP) formulation provides a core imputation. However, unlike in the assignment game, these dual-based imputations, collectively called the \textit{Owen set}, do not fully characterize the core of $b$-matching games (\cite{vazirani2023lpduality}).

The Owen set, first introduced by \cite{Owen1975} for linear production games, was later formalized in \cite{Owen.Characterization}. It consists of core imputations derived from dual solutions, where each dual variable represents an agent's shadow price. More recently, \cite{ggsv2024equitablefaircore} extended the Owen set framework to MST and max-flow games, noting that in these settings, dual variables do not directly correspond to shadow prices, requiring careful interpretation of dual-optimal solutions.

Finally, fairness concepts such as minimax, maximin, leximin, and leximax have been widely explored in cooperative game theory(\cite{minimax_group_fairness}). \cite{Vazirani-leximin} developed an efficient method for computing leximin/leximax fair core imputations for the assignment game. Later, \cite{ggsv2024equitablefaircore} established that computing such imputations for MST and max-flow games is NP-hard, although these imputations can be efficiently computed within the Owen set.



\section{Preliminaries}

\begin{definition}
    Let $N$ be a set of agents. A cooperative game on $N$ is defined by a {\em characteristic function} $\nu: 2^N \rightarrow \mathbb{R}_+$, where for each $S \subseteq N$, $\nu(S)$ is the value that the sub-coalition $S$ can produce on its own. $N$ is also called the grand coalition. We will use $(N,\nu)$ to define a game.
\end{definition}

$\nu(N)$ is also called \textit{value/worth} of the game.  


\begin{definition}
    A \textit{profit share} is $p: N \to \mathbb{R}_+$ is an assignment of profits to the agents.
\end{definition}
An imputation is a kind of profit share that distributes the exact worth of the game among the agents, i.e.,
\begin{definition}
An \textit{imputation} is a profit share $p$ such that $\sum_{i \in N} p(i)=\nu(N)$.
\end{definition}

We extend the definition of $p$ to a set of agents, say $S$, to represent the total profit received by all agents in $S$, i.e., $p(S)=\sum_{i\in S}p(i)$.

\begin{definition}\label{defCore}
An imputation $p$ is in the {\em core} of a profit-sharing game $(N,\nu)$ if and only if for every possible sub coalition of agents $S \subseteq N$, the total profit shares of the members of $S$ is at least as large as the worth they can generate by themselves i.e., $ p(S) \geq \nu(S)$.
\end{definition}

\begin{definition}
Let $P$ be a set of imputations of a game $(N,\nu)$ and $p_1,p_2 \in P$. Let $l_1,l_2$ be the lists formed by arranging the shares of agents in $p_1,p_2$ in ascending order. $l_1$ is {\em lexicographically larger} than $l_2$ if $l_1$ has the larger value at the first index where the two lists differ. The imputation in $P$ which is lexicographically larger than all other imputations in $P$ is the {\em lexicographically minimum} or {\em leximin} imputation in $P$.  
\end{definition}

\begin{definition}
Let $P$ be a set of imputations of a game $(N,\nu)$ and $p_1,p_2 \in P$. Let $l_1,l_2$ be the lists formed by arranging the shares of agents in $p_1,p_2$ in descending order. $l_1$ is {\em lexicographically smaller} than $l_2$ if $l_1$ has the smaller value at the first index where the two lists differ. The imputation in $P$ which is lexicographically smaller than all other imputations in $P$ is the {\em lexicographically maximum} or {\em leximax} imputation in $P$.  
\end{definition}

Let $G = (U,V,E)$ be a weighted bipartite graph with edge weights $w: E \rightarrow \mathbb{R}_+$. Let $b: U\cup V \rightarrow \mathbb{Z}_+$, a vertex capacity function, specify the number of times a vertex can be matched. Any choice of edges, with multiplicity, subject to vertex capacity function, $b$, is called a $b$-matching. 

A \textit{maximum weight bipartite $b$-matching game}, \textit{$b$-matching game} for short, is a game on an instance of bipartite $b$-matching graph with the vertices as agents and the worth of a sub coalition of agents $S\subseteq U\cup V$, denoted by $\nu(S)$, is the weight of a maximum weight $b$-matching, max. wt. $b$-matching for short, in the graph G restricted to vertices in $S$ only. Where obvious, we will use $G$ to mean the set of all agents and $p(G)$ and $\nu(G)$ to mean the total profit of all agents and the worth of the game respectively.


\section{Star Graphs}

In this section, we consider games on \textit{star graphs}, i.e., bipartite graphs where one of the partitions is a singleton set. We will start with the complete characterization of imputations in the core of these games.\\
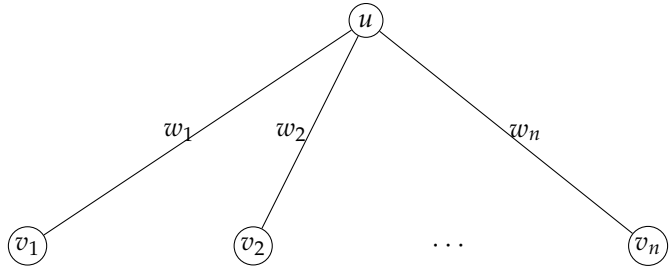
\begin{figure}[H]
    \centering
    \begin{tikzpicture}[scale=1.5, every node/.style={inner sep=1pt}]
        \small
        \node[draw, circle, minimum size=0.45cm] (u) at (0, 2) {$u$};
        \node[draw, circle, minimum size=0.45cm] (v1) at (-3, 0) {$v_1$};
        \node[draw, circle, minimum size=0.45cm] (v2) at (-1, 0) {$v_2$};
        \node[draw, circle, minimum size=0.45cm] (vn) at (2.5, 0) {$v_n$};
        
        \draw (u) -- (v1) node[midway, left] {$w_1$};
        \draw (u) -- (v2) node[midway, left] {$w_2$};
        \draw (u) -- (vn) node[midway, right] {$w_n$};
      
        \path (v2) -- (vn) node[midway, draw=none] {$\cdots$};
    \end{tikzpicture}
    \caption{A star graph $G=(U,V,E)$.}
    \label{fig:general_stars}
\end{figure}
Let us consider a $b$-matching game on a star graph $G=(U,V,E)$ with $U=\{u\}, V=\{v_1,v_2,\ldots, $ $v_n\},$ $E = \{e_i| e_i=(u,v_i), \forall v_i\in V\}$. Let the weight on edges and the capacities on vertices be given by $w:E\rightarrow \mathbb{R}_+$ and $b:U\cup V \rightarrow \mathbb{Z}_+$. We will refer to the agent $u$ as the \textit{central} player and the agents in $V$ as \textit{leaf} players. For ease of notation, let us represent the profit of the central agent by $p_u$ and the leaf agents by $p_i = p_{v_i}$.

\subsection{Efficient Characterization of the Core}

Before we give a characterization, we will first revisit the notion of \textit{marginal utility} of an agent.
\begin{definition}
    In a game $G=(N,\nu)$, the \textit{marginal untility} of an agent $i \in N$ is the decrease in the total worth of the game with the exclusion of $i$.
\end{definition}

We will use $\mu^G(i)$, or $\mu(i)$ if the game is obvious, to represent the marginal utility of agent $i$. Formally, $$ \mu(i) = \nu(G) - \nu(G\setminus \{i\}) $$

\begin{theorem}
\label{thm:star_core_characterization}
    The core of a $b$-matching game on star graphs is completely characterizable. 
\end{theorem}

\begin{proof}
Consider a $b$-matching game on star graph $G$. The core imputations relate to marginal utilities through the following lemma. 
\begin{lemma}
\label{lem:core_marginal_utility}
    An imputation $p:U\cup V \rightarrow \mathbb{R}_+$ is in the core if and only if no leaf player is paid more than its marginal utility, i.e., 
    $$p_i \leq \nu(G) - \nu(G \setminus \{v_i\})$$
\end{lemma}

\begin{proof}
    It is easy to see that the condition is necessary; For contradiction, assume that some leaf player, say $v_i$, gets paid more than its marginal utility, i.e., $p_i > \nu(G) - \nu(G\setminus \{v_i\}) \iff \nu(G) - p_i < \nu(G\setminus \{v_i\})$. Then, by  a simple substitution as shown below, we can see that the sub-coalition of the rest of the vertices, $(U\cup V) \setminus{\{v_i\}}$, will not get enough profit, and so, the imputation is not in the core.  $$p(G\setminus \{v_i\}) = p(G) - p_i = \nu(G) - p_i < \nu(G\setminus \{v_i\})$$

    To show that the condition is also sufficient, we will first discuss a more general property for the star graphs.

    Define $b$-matching ``sub-games'' on star sub-coalitions $S$ of agents, where $S\subseteq U\cup V, u\in S$. We use $S$ to refer to these sub-games, with the characteristic function given by the max. wt. $b$-matchings within this set of agents. Then, we claim the following.

    \begin{claim}
    \label{cl:marginal_utility}
        In star graphs, the marginal utility of a leaf agent is higher in the sub-games than in the original $b$-matching game.\\
        Formally, let $u,v_i \in S$. Then 
    $\mu^S_i \geq \mu^G_i$ or $$ \nu(S) - \nu(S\setminus \{v_i\}) \geq \nu(G) - \nu(G \setminus \{v_i\}) $$
    \end{claim}
    We first argue that the claim is enough to show that the condition is sufficient: Combining $p_i \leq \nu(G)-\nu(G\setminus \{v_i\})$ and the claim, we get $\forall S: p_i \leq \nu(S\cup \{v_i\}) - \nu(S)$. 
    Adding $p(S)$ on both sides to this inequality and rearranging, we have $\forall S: \nu(S\cup \{v_i\}) - p(S\cup \{v_i\}) \geq \nu(S)-p(S)$. Iteratively adding vertices in $G\setminus S$, this inequality gives $\nu(G)-p(G) \geq \nu(S)-p(S)$. Since $p$ is an imputation, we have that $\nu(G)=p(G)$ and so
    
    $$ \forall S, \quad p(S)\geq \nu(S)$$ 

\end{proof}

\end{proof}

Below, we prove the claim on marginal utilities.
\begin{proof}[Proof of~\Cref{cl:marginal_utility}]
    To show this, we show that the marginal utility decreases in the presence of every additional agent, i.e., $$\forall S \subseteq U \cup V, \quad v, v' \notin S: \quad \nu(S \cup \{v\}) - \nu(S) \geq \nu(S \cup \{v, v'\}) - \nu(S \cup \{v'\})$$By adding elements of $G$ not is $S$ in sequence, this identity proves the claim.

   \begin{figure}[H]
    \centering
    \begin{minipage}{0.48\textwidth}
        \begin{tikzpicture}[scale=1, every node/.style={inner sep=1pt}]
            \scriptsize
            \node[draw, circle, minimum size=0.45cm] (u) at (-1, 2) {$u$};
            \node[draw, circle, minimum size=0.45cm] (v1) at (-3, 0) {$v_1$};
            \node[draw, circle, minimum size=0.45cm] (v2) at (-2, 0) {$v_2$};
            \node[draw, circle, minimum size=0.45cm] (vk) at (0, 0) {$v_k$};
            \node[draw, circle, minimum size=0.45cm] (v) at (3, 0) {$v$};

            \draw (u) -- (v1) node[midway, left] {$w_1$};
            \draw (u) -- (v2) node[midway, left] {$w_2$};
            \draw (u) -- (vk) node[midway, right] {$w_k$};
            \draw (u) -- (v) node[midway, right] {$w$};

            \path (v2) -- (vk) node[midway, draw=none] {$\cdots$};

           \draw[blue, thick, rounded corners=25pt] 
            (-4, -0.5) -- (1, -0.5) -- (-1, 2.8) -- cycle;
            \node[blue] at (-1.5, -1) {$S$};
        \end{tikzpicture}
    \end{minipage}
    \hfill
    \begin{minipage}{0.48\textwidth}
        \begin{tikzpicture}[scale=1, every node/.style={inner sep=1pt}]
            \scriptsize
            \node[draw, circle, minimum size=0.45cm] (u) at (-1, 2) {$u$};
            \node[draw, circle, minimum size=0.45cm] (v1) at (-3, 0) {$v_1$};
            \node[draw, circle, minimum size=0.45cm] (v2) at (-2, 0) {$v_2$};
            \node[draw, circle, minimum size=0.45cm] (vk) at (0, 0) {$v_k$};
            \node[draw, circle, minimum size=0.45cm] (v) at (3, 0) {$v$};
            \node[draw, circle, minimum size=0.45cm] (v') at (1, 0) {$v'$};

            \draw (u) -- (v1) node[midway, left] {$w_1$};
            \draw (u) -- (v2) node[midway, left] {$w_2$};
            \draw (u) -- (vk) node[midway, right] {$w_k$};
            \draw (u) -- (v') node[midway, right] {$w'$};
            \draw (u) -- (v) node[midway, right] {$w$};

            \path (v2) -- (vk) node[midway, draw=none] {$\cdots$};

           \draw[blue, thick, rounded corners=25pt] 
            (-4, -0.5) -- (2, -0.5) -- (-1, 2.8) -- cycle;
             \node[blue] at (-1, -1) {$S\cup\{v'\}$};
        \end{tikzpicture}
    \end{minipage}
    \caption{Graphs used in proof of \Cref{cl:marginal_utility}}
    \label{fig:marginal_utility}
\end{figure}
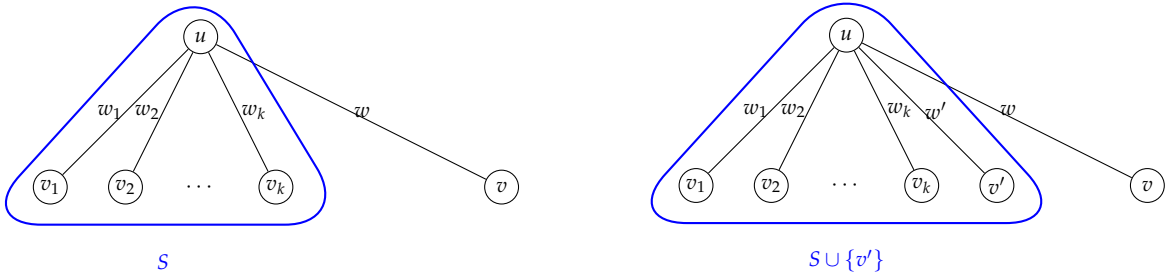

    Consider the graphs in the~\Cref{fig:marginal_utility}. Note that, for star graphs, a max. wt. $b$-matching can just be found by greedily picking the heaviest edges incident on $u$ up to a capacity of $b_{v_i}$ for each $v_i$. 
    
    Let the \textit{weights} of a max. wt. $b$-matching on $S$ be $e_1\geq e_2\geq \ldots \geq e_l$ and that on $S\cup \{v'\}$ be $f_1\geq f_2\geq \ldots \geq f_m$, where $m,l \leq b_u$. Since every edge available in $S$ is also available in $S\cup \{v'\}$, we have $l\leq m$. Let us extend the sequence $e_1\geq e_2\geq \ldots $ up to $m$ terms by adding zeros at the end. As the algorithm is greedy, we also have that $$\forall i\in \{1,2,\ldots,m\} , \quad e_i\leq f_i$$Now bring the vertex $v$ into the graph $S\cup \{v'\}$. Let it replace edges $f_{i_1} \geq f_{i_2} \geq \ldots \geq f_{i_t}$ with the edge $(u,v)$ of weight $w$ to get the max. wt. $b$-matching of $S\cup \{v,v'\}$. Then, replacing the corresponding edges $e_{i_1} \geq e_{i_2} \geq \ldots \geq  e_{i_t}$ with the edge $(u,v)$ will give us a valid $b$-matching in the graph $S\cup \{v\}$. And so, $$\nu(S\cup \{v,v'\}) - \nu(S\cup \{v'\}) = \sum_{j=i}^{t}(w-f_{i_j}) \leq \sum_{j=i}^{t}(w-e_{i_j}) \leq \nu(S\cup \{v\}) - \nu(S)$$ The first inequality comes from $e_i\leq f_i$ and the second inequality is true as the matching obtained by replacing edges $e_{i_1}, e_{i_2},\ldots, e_{i_t}$ with the edge $(u,v)$ need not be the max. wt. $b$-matching in $S\cup \{v\}$. This completes the claim.

 \end{proof}

\begin{corollary}
\label{cr:imputation_star}
There exists a polynomial time algorithm to decide if an imputation is in the core of $b$-matching games on star graphs.
\end{corollary}

\begin{proof}
    From~\Cref{lem:core_marginal_utility}, we only need to compute the marginal utilities of each leaf agent. And marginal utilities of each leaf agent $v_i$ can be computed just by solving for the max. wt. $b$-matching in $G$ and $G\setminus\{v_i\}$. \cite{Schrijver_book} details a strongly polynomial algorithm to compute the max. wt. $b$-matching in a graph. Combining the above, we get a strongly polynomial time algorithm to decide if an imputation is in the core of $b$-matching games on star graphs.
\end{proof}

\subsection{Unstable Coalitions under General Profit Shares}

\begin{definition}
    Given a cooperative game $G=(N,\nu)$ and a profit share $p$, an \textit{unstable coalition} $S\subseteq N$ is a set of agents who receive a profit less than their worth, i.e., $p(S)<\nu(S)$.
\end{definition}

Under this profit share, such sub-coalitions will break away from the grand coalition and hence cause instability. Given a profit share, our goal is to decide if such unstable sub-coalitions exist. The below theorem shows that the task is intractable.

\begin{theorem}
\label{thm:profit_share_for_star}
    Given a $b$-matching game on a star graph and a profit share, deciding if there exists an unstable sub-coalition is NP-complete.
\end{theorem}


\begin{proof}


In this section, we will use the reduction from the \textsc{0-1 Knapsack Problem} which is stated as follows:
\begin{center}
\begin{tabular}{|>{\centering\arraybackslash}m{2cm}|m{10cm}|}
\hline
\multicolumn{2}{|c|}{\textsc{0-1 Knapsack Problem}} \\ \hline
\textbf{Instance:} & A set of $ n $ items $ \{1, 2, \dots, n\} $, each with a weight $ c_i \in \mathbb{N}$ and a value $ a_i \in \mathbb{N} $, and a knapsack with a maximum weight capacity $ C \in \mathbb{N} $ and a goal value $A \in \mathbb{N}$. \\ \hline
\textbf{Question:} & Is there a collection of items $ S \subseteq \{1, 2, \dots, n\} $ such that the total weight $ \sum_{i \in S} c_i \leq C $ and the total value $ \sum_{i \in S} a_i > A$ \\ \hline
\end{tabular}
\end{center}
Consider the following star graph $G^*$.
\begin{figure}[H]
    \centering
    \begin{tikzpicture}[scale=1.5, every node/.style={inner sep=1pt}]
        \small
        \node[draw, circle, minimum size=0.45cm] (u) at (0, 2) {$u$};
        \node[draw, circle, minimum size=0.45cm] (v1) at (-3, 0) {$v_1$};
        \node[draw, circle, minimum size=0.45cm] (v2) at (-1, 0) {$v_2$};
        \node[draw, circle, minimum size=0.45cm] (vn) at (2.5, 0) {$v_n$};
        
        \draw (u) -- (v1) node[midway, left] {$w_1$};
        \draw (u) -- (v2) node[midway, left] {$w_2$};
        \draw (u) -- (vn) node[midway, right] {$w_n$};
      
        \path (v2) -- (vn) node[midway, draw=none] {$\cdots$};
    \end{tikzpicture}
    \caption{Graph $G^*$ used in proof of~\Cref{thm:profit_share_for_star}}
    \label{fig:star}
\end{figure}
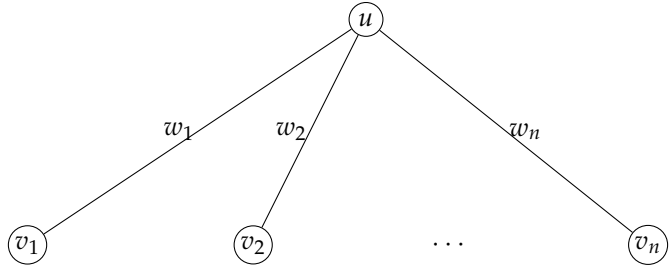







Let $G^*=(U^*,V^*,E^*)$ such that $U^*= \{u\} , V^*=\{v_1,v_2, \ldots,$ $ v_n\}, E^*= \{(u,v_1),(u,v_2), \dots, (u,v_n)\}$. Represent $w: E^* \rightarrow \mathbb{R}_+$, $b: U^*\cup V^* \rightarrow \mathbb{Z}_+$ and $p: U^*\cup V^* \rightarrow \mathbb{R}_+$ as follows.

$$b_i:= b(v_i), \quad w_i:= w((u,v_i)) ,   \quad b_u:=b(u) , \quad p_i:= p(v_i), \quad p_u:=p(u)$$

Construct the instance based on knapsack instance as follows:
$
b_i = c_i, \quad w_i = a_i + 1, \quad p_i = -a_i + c_i(a_i + 1),\quad b_u= C, \quad p_u= A.
$
Let $S \subseteq U^*\cup V^*$ be a sub-coalition then: 
\begin{lemma}
\label{lem:fully_matched_in_star}
If $ v_i \in S $ and $ v_i $ is not fully matched in a max. wt.$ b $-matching in $ S $, then:
$
\nu(S) - p(S) < \nu(S \setminus \{v_i\}) - p(S \setminus \{v_i\}).
$

\end{lemma}

\begin{proof}
Removing $ v_i $ decreases $ p $ by $ p_i $, and as $v_i$ is not fully matched it decreases $ v $ by at most $ (b_i - 1)w_i $ Thus,
$
\nu(S) - p(S) \text{ increases by at least } p_i - (b_i - 1)w_i,
$
which simplifies to $ 1 $. Hence, $ \nu(S) - p(S)$ increases with the addition of $v_i$.

\end{proof}
\begin{corollary}
\label{cor:fully_matched}
    If $S\subseteq U^*\cup V^*$ is a sub-coalition that maximizes $\nu(S) - p(S)$ then every leaf player of $S$ must be fully matched in all max. wt. $b$-matching in $S$.
\end{corollary}
 
We want to check whether all sub-coalitions $S \subseteq U^*\cup V^*$ satisfy $\nu(S) - p(S) \leq 0$. This is true if the condition holds for all sub-coalitions $S$ that maximize the value of $\nu(S) - p(S)$. \Cref{cor:fully_matched} implies that such coalitions must fully match all their leaf players in their max. wt. $b$-matchings. So let us compute the values of $p(S)$ and $\nu(S)$ for these sets.

\begin{lemma}
    \label{lem:value_in_star}
    Assume that in a max. wt. $b$-matching on a sub-coalition of agents $S$, every vertex in $S \cap V^*$ is fully matched. Then the value of this subset is $\sum_{i \in S\cap v} a_i - A$. 
\end{lemma}

\begin{proof}
All leaf vertices are fully matched so for each leaf vertex $v_i\in S$ it increases the value of $S$ by $b_iw_i$ so:
\begin{align*}
\nu(S) & = \sum_{i \in S \cap v} b_i w_i \\
\nu(S) - p(S) & = \sum_{i \in S \cap v} b_i w_i - \bigg(\sum_{i \in S \cap v} p_i + p_u \bigg) \\
 & = \sum_{i \in S \cap v} c_i (a_i + 1) - \bigg(\sum_{i \in S \cap v} (c_i (a_i + 1) - a_i) + A\bigg) \\
 & = \sum_{i \in S \cap v} a_i - A
\end{align*}
 \end{proof}
If there is a sub-coalition of agents $S$ such that $ \nu(S) - p(S) > 0 $, then we have a collection of items in \textsc{0-1 Knapsack Problem} such that the sum of their values is greater than $A$ and the sum of their weights does not exceed the capacity $C$ because the corresponding sub-coalition of this collection makes a valid $b$-matching, so it is satisfying the knapsack constraints. Conversely, any subset $ S' $ satisfying the knapsack constraints corresponds to an unstable coalition$ S $. So the reduction from \textsc{0-1 Knapsack Problem} is complete.
\end{proof}

    \section{General Bipartite Graphs}
\label{sec:general}
In this section, we will show how the NP-completeness result in star graphs results in NP-hardness of recognizing core imputations of general bipartite graphs. We will also show that this, in turn, translates to the NP-hardness of recognizing if an imputation is the leximin(or leximax) core imputation.  

\subsection{Core imputations}

\begin{theorem}
\label{thm:core_coNP_complete}
    Deciding if an imputation is in the core of $b$-matching games on general bipartite graphs is co-NP-Complete.
\end{theorem}


\begin{proof}
Clearly, the problem lies in co-NP, since a co-NP certificate for this problem would be a sub-coalition which is not satisfied under the imputation.\\ 
We will prove that deciding if an imputation is in the core of a bipartite $b$-matching game is co-NP-hard. In this regard, we use a reduction from the problem in~\Cref{thm:profit_share_for_star}. We will construct a bipartite graph $G'$ by adding 2 vertices to a star graph $G$ such that, there is an unstable sub-coalition of agents $S$ in $G'$ if and only if $S$ is an unstable coalition in $G$.


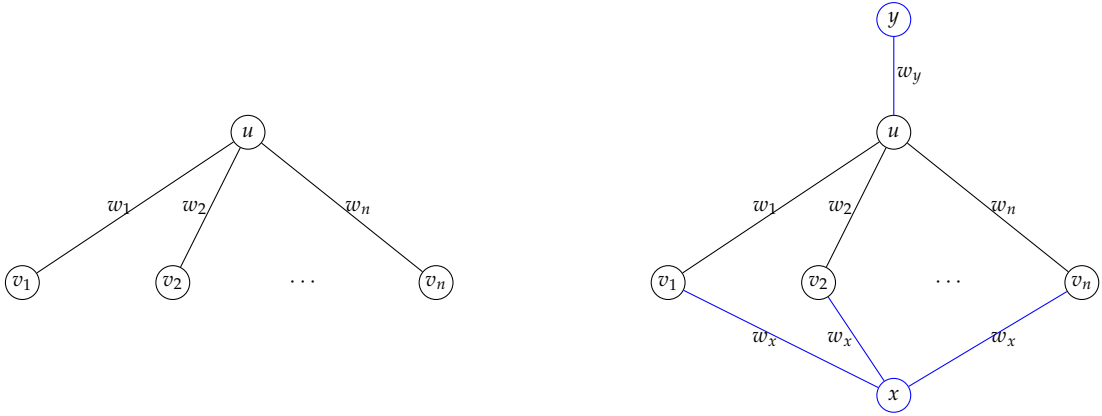
\begin{figure}[H]
    \centering
    \begin{minipage}{0.48\textwidth}
        \begin{tikzpicture}[scale=1, every node/.style={inner sep=1pt}]
            \scriptsize
            \node[draw, circle, minimum size=0.45cm] (u) at (0, 2) {$u$};
            \node[draw, circle, minimum size=0.45cm] (v1) at (-3, 0) {$v_1$};
            \node[draw, circle, minimum size=0.45cm] (v2) at (-1, 0) {$v_2$};
            \node[draw, circle, minimum size=0.45cm] (vn) at (2.5, 0) {$v_n$};

            \draw (u) -- (v1) node[midway, left] {$w_1$};
            \draw (u) -- (v2) node[midway, left] {$w_2$};
            \draw (u) -- (vn) node[midway, right] {$w_n$};
            \path (v2) -- (vn) node[midway, draw=none] {$\cdots$};
        \end{tikzpicture}
    \end{minipage}
    \hfill
    \begin{minipage}{0.48\textwidth}
        \begin{tikzpicture}[scale=1, every node/.style={inner sep=1pt}]
            \scriptsize
            \node[draw, circle, minimum size=0.45cm] (u) at (0, 2) {$u$};
            \node[draw=blue, circle, minimum size=0.45cm] (y) at (0, 3.5) {$y$};
            \node[draw, circle, minimum size=0.45cm] (v1) at (-3, 0) {$v_1$};
            \node[draw, circle, minimum size=0.45cm] (v2) at (-1, 0) {$v_2$};
            \node[draw, circle, minimum size=0.45cm] (vn) at (2.5, 0) {$v_n$};
            \node[draw=blue, circle, minimum size=0.45cm] (x) at (0, -1.5) {$x$};

            \draw (u) -- (v1) node[midway, left] {$w_1$};
            \draw (u) -- (v2) node[midway, left] {$w_2$};
            \draw (u) -- (vn) node[midway, right] {$w_n$};
            \draw[draw=blue] (u) -- (y) node[midway, right] {$w_y$};
            \draw[draw=blue] (v1) -- (x) node[midway, left] {$w_x$};
            \draw[draw=blue] (v2) -- (x) node[midway, left] {$w_x$};
            \draw[draw=blue] (vn) -- (x) node[midway, right] {$w_x$};

            \path (v2) -- (vn) node[midway, draw=none] {$\cdots$};
        \end{tikzpicture}
    \end{minipage}
    \caption{(Left) $G^*;\quad\quad$(Right) Add blue vertices and edges to $G^*$ to get $G$  }
    \label{fig:side_by_side_graph}
\end{figure}

To construct $ G=(U,V,E) $ consider same graph $G^*=(U^*,V^*,E^*)$ of~\Cref{thm:profit_share_for_star} and we add two vertices $x$ and $y$ as follows:
\begin{itemize}
    \item vertex $ x $ is connected to all $ v_i $ with edges of weight $ w_x $.
    \item vertex $ y $ is connected to $ u $ with an edge of weight $ w_y $.
\end{itemize}
So we have $G=(U,V,E)$ such that $U=U^* \cup\{x\}$, $V=V^*\cup\{y\}$, $E=E^*\cup\{(y,u),(x,v_1),$ $(x,v_2),\dots,(x,v_n)\}$

The parameters of the graph are defined as follows:
$$
w_x = \sum_{i=1}^n p_i + 1, \quad w_y = p_u + 1,
$$
$$
b_x = \sum_{i=1}^n b_i, \quad b_y = b_u,
$$
$$
p_x = (b_x - 1)w_x + 1, \quad p_y = (b_y - 1)w_y + 1.
$$

Firstly, we show that $p$ is an imputation in $G$, i.e., we need to show that $p(G) = \nu(G)$. So let us compute these values.

\begin{claim}
    $$p(G) = b_xw_x + b_yw_y$$
\end{claim}

\begin{proof}
The profit $ p(G) $ is the sum of profits assigned to all vertices:
    \begin{align*}
p(G) & = p_x + p_y + p_u + \sum_{i=1}^n p_i.\\
& = (b_x - 1)w_x + 1 + (b_y - 1)w_y + 1 + p_u + \sum_{i=1}^n p_i.\\
 &= b_x w_x - w_x + 1 + b_y w_y - w_y + 1 + p_u + \sum_{i=1}^n p_i\\
 &= b_x w_x - \left( \sum_{i=1}^n p_i + 1 \right) + 1 + b_y w_y - (p_u + 1) + 1 + p_u + \sum_{i=1}^n p_i.\\
\Rightarrow p(G) &= b_x w_x + b_y w_y.
\end{align*}
\end{proof}

\begin{claim}
    $$\nu(G') = b_xw_x+b_yw_y$$
\end{claim}
\begin{proof}
    The value $ \nu(G) $ is the max. wt. $ b $-matching in the graph $ G $. First, we show that in the max. wt. $ b $-matching of $ G $, no edge of the form $ (u,v_i) $ is used.
 
Suppose an edge $ (u,v_i) $ is included in the max. wt. $ b $-matching. We can remove this edge and instead add the edges $ (x,v_i) $ and $ (y,u) $. The value of the game then decreases by $ w_i $, from removing $ (u,v_i) $, but increases by $ w_x + w_y $, from adding $ (x,v_i) $ and $ (y,u) $). Since,
$$
w_x + w_y = \left( \sum_{i=1}^n p_i + 1 \right) + \left( p_u + 1 \right) = p(G) + 2,
$$
and $ p(G) + 2 > w_i $, this contradicts the maximality of the $ b $-matching.

Hence, no edges of the form $ (u,v_i) $ are used in the max. wt. $ b $-matching. Therefore, we can choose all edges adjacent to $x$ and $y$ upto their full capacities, without exceeding the capacities of other vertices to get the max. wt. $b$-matching. Hence, 
$
\nu(G') = b_x w_x + b_y w_y.
$
\end{proof}
\begin{corollary}
    $p(G) = \nu(G)$ so $p$ is an imputation for $G$.
\end{corollary}





\begin{lemma}
    Let $S$ be an unstable coalition containing $x$ or $y$ which are not fully matched in any max. wt. $b$-matching in $S$. Then $S\setminus \{x,y\}$ is also an unstable coalition. 
\end{lemma}

\begin{proof}
Let us consider both cases. \\
\begin{itemize}
\vspace{-6mm}
    \item[] \textbf{Case(\textit{i}): }$v = x$ : 

        Removing $ x $ decreases the profit by $ p_x $ and as $x$ is not fully matched, the value is decreased by at most $ (b_x - 1)w_x $ Thus, $\nu(S) - p(S) \text{ increases by at least } p_x - (b_x - 1)w_x$.

    \item[]\textbf{Case(\textit{ii}): $v = y$} : 
        
        By a similar argument as above we can say that by removing $v_y$, $\nu(S) - p(S)$ increases by at least $ p_y - (b_y - 1)w_y$.
        
\end{itemize}
       
    As we have $p_x - (b_x - 1)w_x = p_y - (b_y - 1)w_y = 1 $, in both cases $\nu(S \backslash\{v\}) - p(S\backslash\{v\}) > \nu(S) - p(S)$. So if $S$ is an unstable coalition that does not fully match $x$ or $y$, then $S\setminus{\{x,y\}}$ is also an unstable coalition. 
\end{proof}

\begin{observation}
    If $S$ is a sub-coalition that maximizes $\nu(S) - p(S)$ and contains $x$ or $y$, then they must be fully matched in every max. wt. $b$-matching of $S$.
\end{observation}
\begin{claim}
    No unstable coalition in $G$ contains $x$ or $y$.
\end{claim}
\begin{proof}

Consider a set $ S $ that maximizes $ \nu(S) - p(S) $. If both $ x $ and $ y $ are in $ S $, they must be fully matched, which implies that all vertices should be included in $ S $. Hence, $ S = G $, and we obtain $ \nu(S) - p(S) = 0 $.

Now, consider the case where $ y $ is in $ S $ but $ x $ is not. Since $ y $ must be fully matched, the edge $ (u,y) $ is chosen to satisfy the capacity of $ u $. Consequently, no edge $ (u,v_i) $ can be matched. Because $ x $ is not in $ S $, the presence of any $ v_i $ does not increase the value of $ S $, and thus these vertices can be removed to maximize $ \nu(S) - p(S) $. Removing all such vertices, we obtain $ S = \{y,u\} $. However, this is not an unstable coalition, as shown by the following computation:
$$p(S) = p_y + p_u = (b_y - 1)w_y + 1 + p_u = (b_u - 1)(p_u + 1) + p_u = b_up_u + b_u $$
$$\nu(S) = b_y.w_y = b_u(p_u+1) = b_up_u + b_u$$

Thus, we have $ p(S) = \nu(S) $.

Next, consider the case where $ x $ is in $ S $ but $ y $ is not. In this scenario, $ x $ must be fully matched, utilizing the entire capacity of the vertices $ v_i $. As a result, these vertices cannot be matched to $ u $, so removing $ u $ further increases $ \nu(S) - p(S) $. This leads to $ S = \{ x,v_1,v_2,\dots,v_n\} $.

$$p(S)= p_x+\sum_{i=1}^n p_i= (b_x - 1)w_x + 1 + \sum_{i=1}^n p_i = (\sum_{i=1}^nb_i - 1)(\sum_{i=1}^np_i + 1) + 1 + \sum_{i=1}^np_i = (\sum_{i=1}^nb_i)(\sum_{i=1}^np_i+1) $$
$$\nu(S)= b_xw_x = (\sum_{i=1}^nb_i)(\sum_{i=1}^np_i+1)  $$

Thus, we again obtain $ p(S) = \nu(S) $.
\end{proof}

We proved that if there is an unstable coalition in $G'$ it has to be an unstable coalition in $G$ so deciding whether imputation $p$ is in the core of $G'$ is equivalent to deciding whether the profit share $p$ is in the core of $G$ which we have proved shown is co-NP-hard. This completes the proof of \cref{thm:core_coNP_complete}.
\end{proof}

\subsection{Fair core imputations}
Using the NP-hardness result above, we will show that finding certain fair core imputations is NP-hard. We will look at finding imputations that maximize the minimum profit or minimize the maximum profit of agents in the $b$-matching game.

\begin{theorem}
    Finding a core imputation that maximizes the minimum profit-share of any vertex in a $b$-matching game is NP-hard. Similarly, it is NP-hard to find a core imputation that minimizes the maximum prove-share of any agent.
\end{theorem}

\begin{proof}
    Consider the game used in \Cref{thm:core_coNP_complete}. It gives a $b$-matching game on graph $G=(U, V, E), w: E\rightarrow \mathbb{R}_+, b:U\cup V\rightarrow \mathbb{Z}_+$ and an imputation $p$ such that, there exists an unstable coalition $S$, if and only if the corresponding items provide a valid knapsack solution.  

    To prove that finding a core imputation that maximizes the minimum profit in a $b$-matching game is NP-hard, we will construct a new graph, $G'$, and provide a new imputation, $p'$, based on $G$ and $p$ such that the $p'$ will give equal profits to all the agents and that $p'$ is in the core of the $G'$ if and only if $p$ is in the core of $G$. Since $p'$ gives equal profit to all agents, if it is in the core, it maximizes the minimum profit of all agents among core imputations. This will imply that recognizing if an imputation is maximizing the minimum profit in a $b$-matching game is co-NP-hard. And as the problem of recognition is co-NP-hard, computation of such a core imputation is NP-hard.
\begin{figure}[H]
    \centering
    \begin{minipage}{0.48\textwidth}
        \begin{tikzpicture}[scale=1, every node/.style={inner sep=1pt}]
            \scriptsize
            \node[draw, circle, minimum size=0.45cm] (u) at (0, 2) {$u$};
            \node[draw=blue, circle, minimum size=0.45cm] (y) at (0, 3.5) {$y$};
            \node[draw, circle, minimum size=0.45cm] (v1) at (-3, 0) {$v_1$};
            \node[draw, circle, minimum size=0.45cm] (v2) at (-1, 0) {$v_2$};
            \node[draw, circle, minimum size=0.45cm] (vn) at (2.5, 0) {$v_n$};
            \node[draw=blue, circle, minimum size=0.45cm] (x) at (0, -1.5) {$x$};

            \draw (u) -- (v1) node[midway, left] {$w_1$};
            \draw (u) -- (v2) node[midway, left] {$w_2$};
            \draw (u) -- (vn) node[midway, right] {$w_n$};
            \draw[draw=blue] (u) -- (y) node[midway, right] {$w_y$};
            \draw[draw=blue] (v1) -- (x) node[midway, left] {$w_x$};
            \draw[draw=blue] (v2) -- (x) node[midway, left] {$w_x$};
            \draw[draw=blue] (vn) -- (x) node[midway, right] {$w_x$};

            \path (v2) -- (vn) node[midway, draw=none] {$\cdots$};
        \end{tikzpicture}
    \end{minipage}
    \hfill
    \begin{minipage}{0.48\textwidth}
        \begin{tikzpicture}[scale=1, every node/.style={inner sep=1pt}]
            \scriptsize
            \node[draw, circle, minimum size=0.45cm] (u) at (0, 2) {$u$};
            \node[draw=blue, circle, minimum size=0.45cm] (y) at (0, 3.5) {$y$};
            \node[draw, circle, minimum size=0.45cm] (v1) at (-3, 0) {$v_1$};
            \node[draw, circle, minimum size=0.45cm] (v2) at (-1, 0) {$v_2$};
            \node[draw, circle, minimum size=0.45cm] (vn) at (2.5, 0) {$v_n$};
            \node[draw=blue, circle, minimum size=0.45cm] (x) at (0, -1.5) {$x$};

            \node[draw=red, circle, minimum size=0.45cm] (u') at (-1, 2) {$u'$};
            \node[draw=red, circle, minimum size=0.45cm] (y') at (-1, 3.5) {$y'$};
            \node[draw=red, circle, minimum size=0.45cm] (v'1) at (-4, 0) {$v'_1$};
            \node[draw=red, circle, minimum size=0.45cm] (v'2) at (-2, 0) {$v'_2$};
            \node[draw=red, circle, minimum size=0.45cm] (v'n) at (1.5, 0) {$v'_n$};
            \node[draw=red, circle, minimum size=0.45cm] (x') at (-1, -1.5) {$x'$};

            \draw (u) -- (v1) node[midway, left] {$w_1$};
            \draw (u) -- (v2) node[midway, left] {$w_2$};
            \draw (u) -- (vn) node[midway, right] {$w_n$};
            \draw[draw=blue] (u) -- (y) node[midway, right] {$w_y$};
            \draw[draw=blue] (v1) -- (x) node[midway, left] {$w_x$};
            \draw[draw=blue] (v2) -- (x) node[midway, left] {$w_x$};
            \draw[draw=blue] (vn) -- (x) node[midway, right] {$w_x$};

            \draw[draw=red] (y) -- (y') node[midway, above] {$w'_y$};
            \draw[draw=red] (u) -- (u') node[midway, above] {$w'_u$};
            \draw[draw=red] (v1) -- (v'1) node[midway, above] {$w'_{v_1}$};
            \draw[draw=red] (v2) -- (v'2) node[midway, above] {$w'_{v_2}$};
            \draw[draw=red] (vn) -- (v'n) node[midway, above] {$w'_{v_n}$};
            \draw[draw=red] (x) -- (x') node[midway, below] {$w'_x$};

            \path (v2) -- (vn) node[midway, draw=none] {$\cdots$};
        \end{tikzpicture}
    \end{minipage}
    \caption{(Left) $G;\quad$(Right) We add partner vertices, shown in red, for every vertex in $G$ to get $G'$}
    \label{fig:extended_graph}
\end{figure}
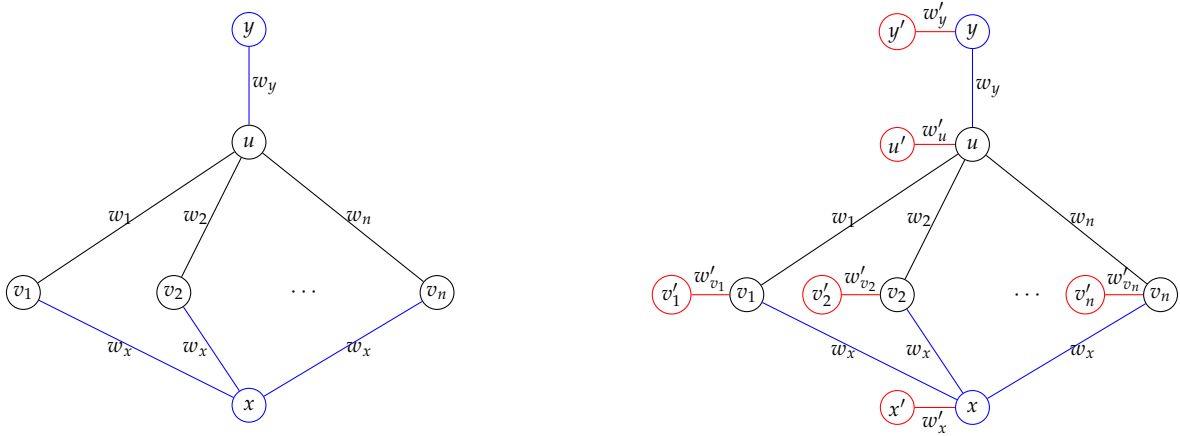

    We will construct $\{G',w',b'\}$ from $\{G,w,b\}$. Let $\nu$ and $\nu'$ represent the characteristic functions of the $b$-matching games in $G$ and $G'$ respectively.
    Construct the graph $G'=(U', V', E')$ from $G=(U, V, E)$ as follows. For every vertex $v$ in $U\cup V$, there are two vertices - $v, v'$ - in $U'\cup V'$. Call $v$'s as the \textit{original} vertices and $v'$'s, the \textit{partner} vertices of respective original vertex $v$.
    
    Preserve all the edges from $G$ in $G'$, i.e., if $(u,v)\in E$, then $(u,v)\in E'$. Weights of all such edges will remain the same. Add an edge between every vertex and its partner, i.e., $\forall v\in U\cup V, (v,v')\in E'$. The weights of these edges will be set based on the profits of agents in $p$ and weights of edges in $E$ as follows. Let $$p^* = 1+ \max{\{\max_{v\in U\cup V} p_v, \max_{e \in E} w_e \}}$$Then $$\forall v \in U\cup V,\quad w'_{(v,v')} = 2p^* - p_v$$ The capacities, $b':U'\cup V'\rightarrow \mathbb{Z}_+$ of all original vertices will be set one more than their capacities in $G$ and the capacities of partner vertices will be set to 1, i.e., $b'_v = 1 + b_v$ and $b'_{v'} = 1$.

    Note that $p^*$ is chosen so that for every vertex, the edge to its partner is heavier than all other edges incident on it - $$ w'_{(v,v')} = 2p^* - p_v = (p^* - p_v) + (p^*) > p^* > w_e,\forall {e \in E}$$ And since we have increased the capacities of all original vertices by 1 and set the capacities of all partner vertices to 1, this, in effect, ensures the following.

    \begin{claim}
    \label{cl:max_wt_b_matching}
        The max. wt. $b$-matching of $G'$ is some max. wt. $b$-matching in $G$ together with all edges between original vertices and their partners. 
    \end{claim}

    \begin{proof}
        Given any $b$-matching, not containing $(v,v')$ for some $v\in U\cup V$, we can include $(v,v')$, replacing any other edge incident on $v$ if necessary, to get another $b$-matching with a larger weight. Hence, every edge between an original vertex and its partner is chosen exactly once in the max. wt. $b$-matching in $G'$. Given these edges are chosen, the capacities on original vertices will become the same as in $G$, completing the proof.  
    \end{proof}

    Finally, consider the profit share $$\forall v\in U'\cup V', \quad p'(v) = p^*$$
    \begin{lemma}
        $p':U'\cup V'\rightarrow \mathbb{R}_+$ is an imputation in the $b$-matching game on $G'$.
    \end{lemma}

    \begin{proof}
        The total profit distributed by $p' = \sum_{v\in U'\cup V'} p^* = 2|(U'\cup V')|\cdot p^*$
        From \Cref{cl:max_wt_b_matching}, the total worth of the game is $\nu(G') = \nu(G)+\sum_{v\in U\cup V}(2p^*-p_v) = \nu(G) + 2|(U\cup V)|\cdot p^* - \sum_{v\in U\cup V} p_v$. Since $p:U\cup V\rightarrow \mathbb{R}_+$ is an imputation in $G$, $\sum_{v\in U\cup V} p_v = p(G) = \nu(G)$, and $\nu(G') = 2|U\cup V|\cdot p^* = p'(G')$. Hence, $p'$ is an imputation in the $b$-matching game on $G'$.
    \end{proof}

    Now we prove the main lemma. 

    \begin{lemma}
    \label{lem:leximin_reduction}
        $p'$ is not in the core of $(G',\nu')$ if and only if $p$ is not in the core of $(G,\nu)$.
    \end{lemma}

    \begin{proof}
        Assume $p$ is not in the core of $(G,\nu)$ as $p(S)<\nu(S)$ for some set $S\subseteq U\cup V$. Let $S'$ be the union of all original vertices and their partner vertices of every vertex in $S$. Then, like in the proof of \Cref{cl:max_wt_b_matching}, $$\nu'(S') = \nu(S) + \sum_{v\in S}(2p^*-p_v) = \nu(S) + 2|S|\cdot p^* - \sum_{v\in S}(p_v) = \nu(S) + 2|S|\cdot p^* - p(S)$$Also note that $p'(S') = 2|S|\cdot p^*$ as every vertex $v$ in $S'$ gets the same profit of $p^*$. Therefore $p(S) < \nu(S) \implies p'(S') < \nu'(S')$. This shows that if $p$ is not in the core of $(G,\nu)$ then $p'$ is not in the core of $(G',\nu')$.

        Now assume $p'$ is not in the core of $(G',\nu')$. Let $S^*\subseteq U'\cup V'$ be a \textit{minimal} set such that $p'(S^*)<\nu'(S^*)$. Note that, for every partner vertex $v'\in S^*$, $S^*$ must also contain the original vertex $v$. For otherwise, we can remove $v'$ from $S^*$ decreasing the profit but not the worth, contradicting the minimality of $S^*$.        
        
        Now, if $S^*$ contains an original vertex $v$ but not its partner $v'$, we can add $v'$ to $S^*$ and still maintain $p'(S^*\cup \{v'\})<\nu'(S^* \cup \{v'\})$ as this increases the value($2p^*-p_v$) more than the profit($p^*$). performing this operation for all original vertices, we will end up with a set $S'$ which has a vertex and its partners in pairs, and $p'(S') < \nu'(S')$.

        Now consider the sub-coalition of agents $S$ of just the original partner vertices in $S'$, in the game $(G,\nu)$. Note that, like above, $p'(S') = 2|S|\cdot p^*$ and $\nu'(S') = \nu(S) + \sum_{v\in S}(2p^*-p_v) = \nu(S) + 2|S|\cdot p^* - \sum_{v\in S}(p_v) = \nu(S) + 2|S|\cdot p^* - p(S)$.

        And so, if $p'(S') < \nu'(S')$ then $p(S) < \nu(S)$ proving the other direction of the lemma.
        
    \end{proof}

    Combining \Cref{lem:leximin_reduction} and \Cref{thm:core_coNP_complete}, we get that finding an unstable coalition in $(G',\nu')$ under imputation $p'$ is NP-hard. Since $p'$ is the imputation that distributes equal profit to all the agents, this proves that it is NP-hard to compute an imputation that maximizes the minimum profit or minimizes the maximum profit. 

\end{proof}

Since leximin and leximax imputations also maximize the minimum profit and minimize the maximum profit, it is also NP-hard to  compute them.

\begin{corollary}
\label{cor:leximin_NP_hard}
    Computing the leximin(or leximax) core imputation in a $b$-matching game is NP-hard.
\end{corollary}    




    
    \bibliographystyle{alpha}
    \bibliography{refs}

\newcommand{\etalchar}[1]{$^{#1}$}
\begin{thebibliography}{VGPR{\etalchar{+}}00}

\bibitem[Bir76]{Bird1976}
Charles~G Bird.
\newblock On cost allocation for a spanning tree: a game theoretic approach.
\newblock {\em Networks}, 6(4):335--350, 1976.

\bibitem[BKPW18]{biro2018stable}
P{\'e}ter Bir{\'o}, Walter Kern, Dani{\"e}l Paulusma, and P{\'e}ter Wojuteczky.
\newblock The stable fixtures problem with payments.
\newblock {\em Games and economic behavior}, 108:245--268, 2018.

\bibitem[BVP07]{bergantinos2007fair}
Gustavo Berganti{\~n}os and Juan~J Vidal-Puga.
\newblock A fair rule in minimum cost spanning tree problems.
\newblock {\em Journal of Economic Theory}, 137(1):326--352, 2007.

\bibitem[CCPS11]{Schrijver_book}
W.J. Cook, W.H. Cunningham, W.R. Pulleyblank, and A.~Schrijver.
\newblock {\em Combinatorial Optimization}.
\newblock Wiley Series in Discrete Mathematics and Optimization. Wiley, 2011.

\bibitem[DGK{\etalchar{+}}21]{minimax_group_fairness}
Emily Diana, Wesley Gill, Michael Kearns, Krishnaram Kenthapadi, and Aaron Roth.
\newblock Minimax group fairness: Algorithms and experiments.
\newblock In {\em Proceedings of the 2021 AAAI/ACM Conference on AI, Ethics, and Society}, AIES '21, page 66–76, New York, NY, USA, 2021. Association for Computing Machinery.

\bibitem[FKFH97]{Faigle1997complexity}
Ulrich Faigle, Walter Kern, S{\'a}ndor~P Fekete, and Winfried Hochst{\"a}ttler.
\newblock On the complexity of testing membership in the core of min-cost spanning tree games.
\newblock {\em International Journal of Game Theory}, 26:361--366, 1997.

\bibitem[FTM94]{feltkamp1994irreducible}
V.~Feltkamp, S.H. Tijs, and S.~Muto.
\newblock On the irreducible core and the equal remaining obligations rule of minimum cost spanning extension problems.
\newblock Workingpaper, Unknown Publisher, 1994.
\newblock Pagination: 36.

\bibitem[FZCD02]{Fang2002computational}
Qizhi Fang, Shanfeng Zhu, Maocheng Cai, and Xiaotie Deng.
\newblock On computational complexity of membership test in flow games and linear production games.
\newblock {\em International Journal of Game Theory}, 31:39--45, 2002.

\bibitem[GGSV24]{ggsv2024equitablefaircore}
Rohith~R. Gangam, Naveen Garg, Parnian Shahkar, and Vijay~V. Vazirani.
\newblock Leximin and leximax fair core imputations for the max-flow, mst, and bipartite $b$-matching games, 2024.

\bibitem[GH81]{GranotHuberman1981}
Daniel Granot and Gur Huberman.
\newblock Minimum cost spanning tree games.
\newblock {\em Mathematical programming}, 21:1--18, 1981.

\bibitem[Kar02]{Kar}
Anirban Kar.
\newblock Axiomatization of the shapley value on minimum cost spanning tree games.
\newblock {\em Games and Economic Behavior}, 38(2):265--277, 2002.

\bibitem[KTZ24]{b_matching_nucleolus}
Jochen Könemann, Justin Toth, and Felix Zhou.
\newblock On the complexity of nucleolus computation for bipartite b-matching games.
\newblock {\em Theoretical Computer Science}, 998:114476, 2024.

\bibitem[KZ82]{Kalai1982totally}
Ehud Kalai and Eitan Zemel.
\newblock Totally balanced games and games of flow.
\newblock {\em Mathematics of Operations Research}, 7(3):476--478, 1982.

\bibitem[Owe75]{Owen1975}
Guillermo Owen.
\newblock On the core of linear production games.
\newblock {\em Mathematical programming}, 9(1):358--370, 1975.

\bibitem[Sot92]{sotomayor1992multiple}
Marilda Sotomayor.
\newblock The multiple partners game.
\newblock In {\em Equilibrium and Dynamics: Essays in Honour of David Gale}, pages 322--354. Springer, 1992.

\bibitem[SS71]{Shapley1971assignment}
Lloyd~S Shapley and Martin Shubik.
\newblock The assignment game {I}: The core.
\newblock {\em International Journal of Game Theory}, 1(1):111--130, 1971.

\bibitem[SSLGJ01]{Transportation_games}
Joaqu\'{\i}n S{\'a}nchez-Soriano, Marco~A. L{\'o}pez, and Ignacio Garc\'{\i}a-Jurado.
\newblock On the core of transportation games.
\newblock {\em Mathematical Social Sciences}, 41(2):215--225, 2001.

\bibitem[Tru12]{trudeau2012new}
Christian Trudeau.
\newblock A new stable and more responsive cost sharing solution for minimum cost spanning tree problems.
\newblock {\em Games and Economic Behavior}, 75(1):402--412, 2012.

\bibitem[Vaz23]{vazirani2023lpduality}
Vijay~V. Vazirani.
\newblock {LP-D}uality theory and the cores of games.
\newblock {\em arXiv preprint arXiv: 2302.07627}, 2023.

\bibitem[Vaz24]{Vazirani-leximin}
Vijay~V Vazirani.
\newblock Equitable core imputations via a new adaptation of the primal-dual framework.
\newblock {\em arXiv preprint arXiv:2402.11437}, 2024.

\bibitem[VGPR{\etalchar{+}}00]{Owen.Characterization}
JRG Van~Gellekom, Jos~AM Potters, JH~Reijnierse, MC~Engel, and SH~Tijs.
\newblock Characterization of the owen set of linear production processes.
\newblock {\em Games and Economic Behavior}, 32(1):139--156, 2000.

\end{thebibliography}
    
    
\end{document}